\documentclass[10pt,conference,twocolumn]{IEEEtran}
\usepackage{amsmath}
\usepackage{times}
\usepackage{algorithm}
\usepackage[noend]{algorithmic}
\usepackage{algorithmic}
\usepackage{graphicx}
\usepackage{verbatim}
\usepackage{amssymb}
\usepackage{cite}
\usepackage{epsfig}
\usepackage[pdf]{pstricks}
\usepackage[crop=off]{auto-pst-pdf}
\usepackage{caption}
\usepackage{subcaption}
\usepackage{color}

\DeclareMathOperator*{\argmax}{arg\,max}
\DeclareMathOperator*{\argmin}{arg\,min}

\newcommand\txtblue[1]{{\color{black}#1}}

\newtheorem{lemma}{Lemma}

\newtheorem{proposition}{Proposition}
\newtheorem{corollary}{Corollary}

\title{Asymptotically\txtblue{-}Optimal Incentive-Based \txtblue{En-route} Caching Scheme}

\author{\IEEEauthorblockN{Ammar Gharaibeh$^{\dag}$, Abdallah Khreishah$^{\dag}$, Issa Khalil$^{\star}$, Jie Wu$^{\diamond}$
}
\IEEEauthorblockA{$^{\dag}$New Jersey Institute of Technology, $^{\star}$Qatar Computing Research Institute, $^{\diamond}$Temple University\\
amg54@njit.edu, abdallah@njit.edu, ikhalil@qf.org.qa, jiewu@temple.edu}
}

\begin{document}
\maketitle

\begin{abstract}
Content caching at intermediate nodes is a very effective way to optimize the operations of Computer networks, so that future requests can be served without going back to the origin of the content. Several caching techniques have been proposed since the emergence of the concept, including techniques that require major changes to the Internet architecture such as Content Centric Networking. Few of these techniques consider providing caching incentives for the nodes or quality of service guarantees for content owners. In this work, we present a low complexity, distributed, and online algorithm for making caching decisions based on content popularity, while taking into account the aforementioned issues. Our algorithm performs en-route caching. Therefore, it can be integrated with the current TCP/IP model. In order to measure the performance of any online caching algorithm, we define the competitive ratio as the ratio of the performance of the online algorithm in terms of traffic savings to the performance of the optimal offline algorithm that has a complete knowledge of the future. We show that under our settings, no online algorithm can achieve a better competitive ratio than $\Omega(\log n)$, where $n$ is the number of nodes in the network. Furthermore, we show that under realistic scenarios, our algorithm has an asymptotically optimal competitive ratio in terms of the number of nodes in the network. We also study an extension to the basic algorithm and show its effectiveness through extensive simulations.
\end{abstract}


\section{Introduction}\label{sec:Introduction}
Recently, content retrieval has dominated the Internet traffic. Services like Video on Demand accounts for 53\% of the total Internet traffic, and it is expected to grow even further to 69\% by the end of 2018 \cite{cisco2012forecast}. Content Delivery Network (CDN) uses content replication schemes at \txtblue{dedicated servers} to bring the contents closer to the requesting customers. This has the effect of offloading the traffic from the origin servers, reducing content delivery time, and achieving better performance, scalability, and energy efficiency \cite{vakali2003content,pathan2007taxonomy}. Akamai, for example, is one of the largest CDNs deployed, delivering around 30\% of web traffic through globally-distributed platforms \cite{nygren2010akamai}. \txtblue{The problem with CDN is the necessity of \emph{dedicated servers} and that content replication is done offline.}

\txtblue{Several techniques have emerged to overcome the limitation of caching at dedicated servers. For example, }Content Centric Networking (CCN) \cite{jacobson2009networking} uses the content name instead of the IP address of the source to locate the content. This allows more flexible caching at intermediate nodes. In order to implement CCN, major changes in the TCP/IP protocol needs to be performed. When a client requests certain content, the client sends an \emph{Interest Packet} to all its neighbors, which in turn send the packet to all of their neighbors except the one where the packet came from. The process continues until a node caching the desired content is found, which in turn replies with a \emph{Data Packet} containing the desired content.

Clearly, caching a content will reduce the traffic on the upstream path, if the same content is being requested another time by a different client. Given the limited cache capacity, the questions to answer become \lq What are the factors that affect achieving the maximum traffic savings?\rq and \lq Which contents are to be cached in order to achieve the same objective?\rq

Several studies try to answer the above questions. The work in \cite{psaras2011modelling} investigates the dependence of the caching benefit on content popularity, nodes' caching capacities, and the distance between nodes and the origin server. The performance of CCN has been evaluated in \cite{rossi2011caching} under different topologies, by varying routing strategies, caching decisions, and cache replacement policies. The results also show the dependence of CCN performance on content popularity.

Several techniques for content caching have been proposed in the literature. The work in \cite{jacobson2009networking} presents \emph{Always Cache}, where a node caches every new piece of content under the constraint of cache capacity. The authors in \cite{guan2013push} provide a push-pull model to optimize the joint latency-traffic problem by deciding which contents to push (cache) on intermediate nodes, and which contents to pull (retrieve) from the origin server. Most Popular Caching caches a content at neighboring nodes when the number of requests exceeds some threshold \cite{bernardini2013mpc}. \emph{ProbCache} aims to reduce the cache redundancy by caching contents at nodes that are close to the destination \cite{psaras2012probabilistic}. A cooperative approach in \cite{fiore2009cache} leads to a node's caching decision that depends on its estimate of what neighboring nodes have in their cache. A collaborative caching mechanism in \cite{dai2012collaborative} maximizes cache cooperation through dynamic request routing. In \cite{laoutaris2007cache}, nodes try to grasp an idea of other nodes' caching policies through requests coming from those nodes.

Few works targeted the caching decision problem from the point of view of optimality, or providing incentives for nodes to cache. \txtblue{The work in \cite{jiang2003optimal} presents an offline solution through dynamic programming for content placement for en-route caching}. Authors in \cite{llorca2013dynamic} characterize the optimal content placement strategy under offline settings, in which all future requests are known to all nodes in the network. The work of \cite{rosensweig2009breadcrumbs} presents an online solution but with no efficiency or optimality proofs. Other works such as \cite{rajahalme2008incentive} and \cite{pham2013pricing} consider incentives for nodes to cache. However, they provide high level solutions that do not scale well with large systems. The authors in \cite{pham2013pricing} consider a special case with only 3 ISPs.

This paper provides a provably-optimal online solution for the first time under a setting that brings incentives for the nodes to cache. In order to provide incentives for the nodes to cache, nodes have to charge content providers for caching their contents. Adopting such charging policies forces the caching node to provide quality of service guarantees for content providers by not replacing their contents in the future, if the node decides to cache their contents. Since the number of contents far exceeds the nodes' cache capacities, and assuming that the charging price for every piece of content is the same, then the node has no preference in caching one content over the other, forcing the node to cooperate and apply our policy that achieves \txtblue{asymptotic} optimality.

Specifically, we make the following contributions:

{\bf(1)} We design an online, low complexity, and distributed caching decision algorithm that provides incentives for the nodes to cache, and quality of service guarantees for content providers.
{\bf(2)} Our algorithm performs en-route caching and thus can be implemented without radical changes to the TCP/IP protocol stack.
{\bf(3)} Under some realistic network settings, We show that our algorithm is asymptotically (in terms of the number of nodes in the network) optimal (in terms of traffic savings). 
{\bf(4)} Through extensive simulations, we show that our algorithm outperforms existing caching schemes. We also show the effeciency of an extension of our algorithm with respect to the existing caching schemes.

The rest of the paper is organized as follows: Section \ref{sec:Settings} states the definitions and settings of our algorithm. Section \ref{sec:Algorithm} describes the algorithm and practical issues. Optimality analysis of the algorithm is presented in Section \ref{sec:Performance}. Section \ref{sec:Extensions} describes the extensions of our algorithm. Section \ref{sec:Simulation} provides simulation results. We conclude the paper in Section \ref{sec:Conclusion}.

\section{Settings and Definitions}\label{sec:Settings}
In this Section, we provide the settings under which our algorithm takes place, followed by some necessary definitions.

\subsection{Settings}
A network is represented by a graph $G (V,E)$, where each node $i\in V$ has a caching capacity of $D_{i}$. If the node does not have caching capability, its caching capacity is set to 0. Weights can be assigned to each link $e \in E$, but we consider all links to have the same weight. The input consists of a sequence of contents $\beta_{1}, \beta_{2},...,\beta_{m}$, the $j$-th of which is represented by $\beta_{j} = (S_{j}, r_{j}, T_{j}(\tau))$, where $S_{j}$ is the source for content $\beta_{j}$, $r_{j}$ is the size of $\beta_{j}$, and $T_{j}(\tau)$ is the effective caching duration in which more requests are expected for $\beta_{j}$ when a request appears at time slot $\tau$. For simplicity, we assume a slotted time system and that $T_{j}(\tau)$ is an integer multiple of slots.

For each content, we define the following values:

{\bf(1)} $b_{i}(j)$: Number of hops on the path from node $i$ to $S_{j}$ for $\beta_{j}$.

{\bf(2)} $W_{i}(\tau,j)$: The expected number of requests for $\beta_{j}$ to be served from the cache at node $i$ at time slot $\tau$, if all of the caching nodes cache $\beta_{j}$.

{\bf(3)} $t_{0}(i,j)$: The time when a request for $\beta_{j}$ appears at node $i$.

{\bf(4)} $\mathcal{E}_{i}(\tau,j)$: The total expected number of requests for $\beta_{j}$ to be served from the cache at node $i$ per time slot $\tau$. We assume that $\mathcal{E}_{i}(\tau,j)$ is fixed $\forall \tau \in \{t_{0},\ldots,t_{0}+T_{j}(t_{0})\}$.

{\bf(5)} $\tau_{0}(i,j)$: The time when $\beta_{j}$ is cached at node $i$. For simplicity, we denote this value hereafter by $\tau_{0}$ since the values of $(i,j)$ can be inferred from the context.

{\bf(6)} $d_{i}(\tau,j)$: Number of hops from node $i$ to the first node caching $\beta_{j}$ along the path to $S_{j}$ at time $\tau$. We assume that if node $i$ caches $\beta_{j}$ at time $\tau_{0}$, then $d_{i}(\tau,j) = d_{i}(\tau_{0},j), \forall \tau \in \{\tau_{0},\ldots,\tau_{0} + T_{j}(\tau_{0})\}$.

Figure \ref{fig:Simple} shows a simple network to illustrate the aforementioned definitions. In this example, we have two contents $\beta_{1}$ and $\beta_{2}$, originally stored on $v_{1}$ and $v_{2}$, respectively. The triangles in the figure represent the subnetworks containing the set of non-caching nodes connected to the caching node. The values of $W_{i}(\tau,j)$ represent the expected number of requests for $\beta_{j}$ coming from the set of non-caching nodes in the subnetwork connected to node $i$.

Before any requests for $\beta_{j}$ appears at any node, each node $i$ will send its $W_{i}(\tau,j)$ to all nodes on the path from node $i$ to the source of $\beta_{j}$, $S_{j}$. This process will lead to the calculation of the initial values of $\mathcal{E}_{i}(\tau,j)$.

For example, in Figure \ref{fig:Simple}, before any request for $\beta_{1}$ appears at any node, $\mathcal{E}_{3}(\tau,1) = W_{3}(\tau,1) + W_{4}(\tau,1)$, to a total value of 6. This is because, starting from the initial configuration while investigating the caching of content $\beta_{1}$ on node $v_{3}$, all the requests for $\beta_{1}$ coming from the subnetworks connected to $v_{3}$ and $v_{4}$ will be served from the cache of $v_3$, if we decide to cache $\beta_{1}$ on $v_{3}$. Similarly, $\mathcal{E}_{2}(\tau,1) = 9$. Later on, if $v_{4}$ decides to cache $\beta_{1}$, then $W_{4}(\tau,1)$ will be subtracted from all nodes along the path to $S_{1}$, until the first node caching $\beta_{1}$ is reached. This is because none of these nodes will serve the requests for $\beta_{1}$ coming from the subnetwork connected to $v_{4}$ after this point. In Sections \ref{sec:Algorithm} and \ref{subsec:expec}, we provide details for the dynamic calculation and initialization of $\mathcal{E}_{i}(\tau,j)$, respectively.

\begin{figure}
\centering
\begin{minipage}{.45\linewidth}
\includegraphics[width=\linewidth]{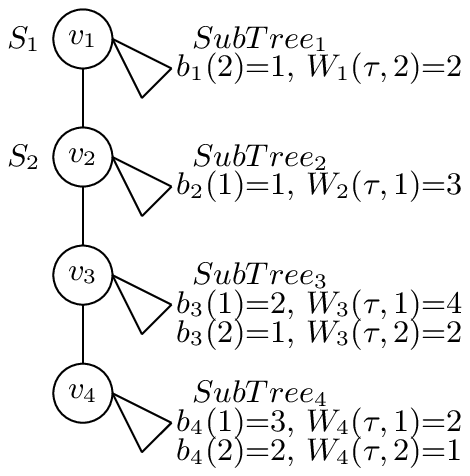}
\captionof{figure}{Simple Caching Network.}
\label{fig:Simple}
\end{minipage}
\hspace{.05\linewidth}
\begin{minipage}{.45\linewidth}
\centering
\includegraphics[width=\linewidth]{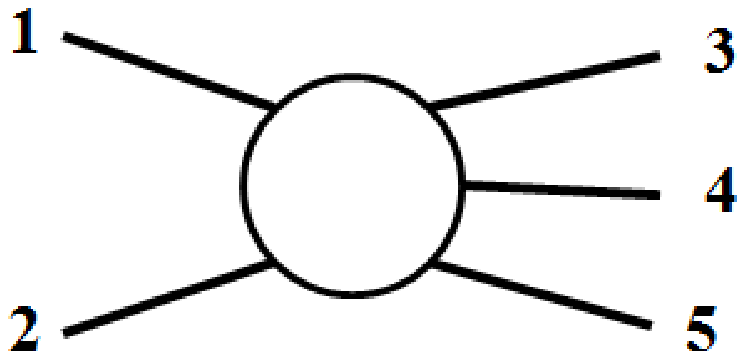}
\captionof{figure}{A single node in CCN.}
\label{fig:hier}
\end{minipage}
\end{figure}

We define the total traffic savings of caching in the time interval [0,$t$] as:
\begin{equation}
\sum_{\tau = 0}^{t}\sum_{i = 1}^{n}\sum_{j = 1}^{m}\mathcal{E}_{i}(\tau_{0},j)d_{i}(\tau_{0},j)I(a_{i}(\tau,j)),
\label{eq:savings}
\end{equation}
where $I(.)$ is the indicator function and $a_{i}(\tau,j)$ is the event that $\beta_{j}$ exists at node $i$ at time $\tau$. For example, referring to Figure \ref{fig:Simple}, caching $\beta_{1}$ on $v_{3}$ alone for a single time slot will yield a saving of $\mathcal{E}_{3}(\tau,1)\times d_{3}(\tau,1) = (4 + 2)\times2 = 12$.

We define the relative load on a caching node $i$ at time $\tau$ when $\beta_{j}$ arrives as
\begin{displaymath}
\lambda_{i}(\tau,j) = \sum_{\substack{k: k < j\\
k \in Cache_{i}(\tau)}} \frac{r_{k}}{D_{i}},
\end{displaymath}
where $k<j$ refers to the indices of all $\beta_{k}$ that are in the cache of node $i$ at the time when considering $\beta_{j}$ to be cached at node $i$. We use $k \in Cache_{i}(\tau)$ to represent the existence of $\beta_{k}$ in the cache of node $i$ at time $\tau$.

As we mentioned in Section \ref{sec:Introduction}, charging content providers for caching their contents will provide the nodes with the necessary incentives to cache. In return, the nodes have to guarantee quality of service for content providers by keeping their content cached for the required time period. We assume that content providers are charged the same to prevent the node from prefering contents with a higher prices. To this end, we consider \emph{non-preemptive} caching to represent our system model, \emph{i.e.}, once $\beta_{j}$ is cached at node $i$, it will stay cached $\forall \tau \in \{\tau_{0},\ldots,\tau_{0}+T_{j}(\tau_{0})\}$ time slots. We elaborate more on $T_{j}(\tau)$ in Section \ref{subsec:Timing}.

\subsection{Definitions}
\emph{Offline vs. Online Algorithms}: The main difference between the offline and the online algorithms is that the offline algorithm has a complete knowledge of the future. In our work, offline means that the algorithm knows \emph{when}, \emph{where}, and \emph{how many times} a content will be requested. This knowledge leads to the optimal content distribution strategy that maximizes the performance in terms of traffic savings. On the other hand, online algorithms do not possess such knowledge. Online algorithms have to make a caching decision for a content based on the available information at the time of the content arrival. Due to this difference, the offline algorithm's performance is better than that of the online algorithm.

Under our settings, we assume that the node does not know when a request for a content will come. However, once a request for a content arrives at a caching node, the node will know the content's size, the effective caching duration time, and the expected number of requests to be served from the cache of the caching node. Furthermore, all other caching nodes are informed about the arrival time of the request. We elaborate more on this issue in Section \ref{subsec:Timing}. For example, referring back to Figure \ref{fig:Simple}, node $v_{3}$ does not know when a request for $\beta_{1}$ will come. Only when a request for $\beta_{1}$ arrives at $v_{3}$ \txtblue{at time $t_{0}$}, does $v_{3}$ know $r_{1}, T_{1}(t_{0}), \mathcal{E}_{3}(\tau,1)$, in addition to its own relative load, $\lambda_{3}(\tau,1), \txtblue{\forall \tau \in \{t_{0},\ldots,t_{0}+T_{1}(t_{0})\}}$. However, node $v_{3}$ does not know when the next request for the same content will come.

To measure the performance in terms of \emph{traffic savings}, as defined in \eqref{eq:savings}, of the online algorithm against the offline algorithm, we use the concept of \emph{Competitive Ratio}. Here, traffic savings refer to, but not limited, to the total number of hops saved using \txtblue{en-route caching}, compared to the traditional no-\txtblue{caching} case in which the request for a content is served by the content's source. The traffic savings can be based on other metrics like the actual distance or the energy consumption. Other works have used the concept of competitive ratio, but for different problems such as energy efficiency \cite{albers2007energy} or online routing \cite{jaillet2008generalized}. Competitive ratio is defined as the performance achieved by the offline algorithm to the performance achieved by the online algorithm, \emph{i.e.}, if we denote the offline performance as $P_{off}$ and the online performance as $P_{on}$, the competitive ratio is:
\begin{displaymath}
\sup_{t} \sup_{\substack{all\ input\\
sequences\ in\ [0,t]}} \frac{P_{off}}{P_{on}}.
\end{displaymath}
As the ratio gets closer to 1, the online performance gets closer to the offline performance. In other words, the smaller the competitive ratio, the better the online algorithm's performance.

We motivate the design of our online algorithm by the following reasoning; knowing the contents' popularities alone does not guarantee an optimal solution. The order in which the contents arrive makes a big difference.

Referring to Figure \ref{fig:Simple}, consider the existence of two contents named $X$ and $Y$, originally located at $v_{1}$. Assume that $W_{3}(\tau,X)=W_{4}(\tau,X)=1$, $W_{3}(\tau,Y)=1$, and $W_{4}(\tau,Y)=10$. Assume that all nodes have enough residual cache capacity for one content except node $v_{3}$, which is full and cannot cache any content. Furthermore, assume that $X$ and $Y$ will be both requested twice by $v_{4}$ at different time slots. Consider the following two scenarios:

{\bf(En-Route Caching)}: If the first request for content $X$, followed by the first request for content $Y$, arrives at $v_{4}$, then $v_{4}$ will cache the first content $X$ and $v_{2}$ will cache content $Y$, achieving a traffic saving at $v_{4}$ for the next pair of requests of $(1\times3) + (10\times1) = 13$. Later on, if requests for $X$ and $Y$ appear at $v_{3}$, then $v_{3}$ will get content $X$ from $v_{1}$ and content $Y$ from $v_{2}$, gaining an additional savings of $(1\times0) + (1\times1) = 1$.

On the other hand, if a request for $Y$ is followed by a request for $X$ at $v_{4}$, then $v_{4}$ will cache the first content $Y$ and $v_{2}$ will cache content $X$, achieving a traffic saving at $v_{4}$ for the next pair of requests of $(10\times3) + (1\times1) = 31$. Later on, if requests for $X$ and $Y$ appear at $v_{3}$, then $v_{3}$ will get content $X$ from $v_{2}$ and content $Y$ from $v_{1}$, gaining an additional savings of $(1\times1) + (1\times0) = 1$. So the online algorithm will achieve an average traffic savings of $23$.

The offline algorithm knows in advance that content $Y$ will be requested and can reject the caching of content $X$ at $v_{4}$ and cache it at $v_{2}$ to achieve a traffic saving of $(10\times3) + (1\times1) + (1\times1) + (1\times0) = 32$.

{\bf(Routing to the Closest Caching Node)}: If the first request for content $X$, followed by the first request for content $Y$, arrives at $v_{4}$, then $v_{4}$ will cache the first content $X$, and $v_{2}$ will cache content $Y$, achieving a traffic saving at $v_{4}$ for the next pair of requests of $(1\times3) + (10\times1) = 13$. Later on, if requests for $X$ and $Y$ appear at $v_{3}$, then $v_{3}$ will get content $X$ from $v_{4}$ and content $Y$ from $v_{2}$, gaining an additional savings of $(1\times1) + (1\times1) = 2$.

On the other hand, if a request for $Y$ is followed by a request for $X$ at $v_{4}$, then $v_{4}$ will cache the first content $Y$ and $v_{2}$ will cache content $X$, achieving a traffic saving at $v_{4}$ for the next pair of requests of $(10\times3) + (1\times1) = 31$. Later on, if requests for $X$ and $Y$ appear at $v_{3}$, then $v_{3}$ will get content $X$ from $v_{2}$ and content $Y$ from $v_{4}$, gaining an additional savings of $(1\times1) + (1\times1) = 2$. Because of this, the online algorithm will achieve an average traffic savings of $24$.

The offline algorithm knows in advance that content $Y$ will be requested and can reject the caching of content $X$ at $v_{4}$, and will cache it at $v_{2}$ to achieve a traffic saving of $(10\times3) + (1\times1) + (1\times1) + (1\times1) = 33$.

The above examples show that the online algorithm cannot guarantee an optimal solution. In fact, we show that there is an upper bound on the savings achieved by the online algorithm when compared to the offline algorithm, and we develop an online algorithm that achieves that bound under realistic settings.

\section{Algorithm}\label{sec:Algorithm}
In this Section, we present the Cost-Reward Caching (CRC) algorithm that achieves the optimal competitive ratio, along with some practical issues. We introduce the proof of optimality in the next Section.
\subsection{CRC Algorithm}
CRC takes advantage of \txtblue{en-route caching}, \emph{i.e.}, a request for a content is forwarded \txtblue{along the path to the content's source}, up to the first node that has the content in its cache. The content then will follow the same path back to the requester.

In CCN, when an interest packet for a new content arrives at a node on a certain interface, the node will send the interest packet using all other interfaces. For example, Figure \ref{fig:hier} shows a single node in CCN, where the numbers represent the interfaces of the node. When a request for $\beta_{j}$ arrives at the node through interface number 2, and a match is not found in neither the cache nor the Pending Interest Table (PIT), the node will send the request on all interfaces except interface number 2. Our algorithm uses \txtblue{en-route caching}, so the new interest packet is only forwarded on the single interface along the path to the content's source.

When a request for a content $\beta_{j}$ appears at a node $i$ \txtblue{at time $t_{0}$, node $i$ sends a small control message up to the first node caching $\beta_{j}$ along the path to the source of the content. Let $w$ be that first node, then node $w$ replies with a message containing $r_j$ and the ID of node $w$. Every node $u$ in the path from node $w$ to node $i$ stores a copy of the message, computes $d_{u}(t_0,j)$, and forwards the message to the next node along the path to node $i$}. When Node $i$ recieves the message, it makes a caching decision according to Algorithm \ref{alg:alg2}. If node $i$ decides to cache $\beta_{j}$, it initializes a header field in the request packet to the value of $\mathcal{E}_{i}(\tau,j)$. If node $i$ decides not to cache, it initializes the header field to 0.

The request packet is then forwarded to the parent node $z$. The parent first subtracts the value stored in the header field from its own value of $\mathcal{E}_{z}(\tau,j)$. Based on the new value of $\mathcal{E}_{z}(\tau,j)$, if node $z$ decides to cache $\beta_{j}$, it adds its $\mathcal{E}_{z}(\tau,j)$ to the value in the header field. Otherwise, node $z$ adds 0. The request packet is then forwarded to node $z$'s parent, and the whole process is repeated until the request reaches the first node that has the content in its cache. The content then will follow the same path back to the requester, and every node in the path that decided to cache the content will store a copy in its cache. We describe the operation of our algorithm in Algorithm \ref{alg:alg1}. 

\begin{algorithm}
\caption{En-Route Caching}
\label{alg:alg1}
\begin{algorithmic}
\STATE{A request for $\beta_{j}$ appears at node $i$ at time $t_{0}$.}
\STATE{$header = 0$}
\IF{$\beta_{j} \in Cache_{i}(t_{0})$}
\STATE{Reply back with $\beta_{j}$}
\ELSE
\STATE{Send a control message to retrieve $r_{j}$, $d_{i}(t_{0},j)$}
\STATE{$w \leftarrow$ first node on the path to $S_{j}$, where $\beta_{j} \in Cache_{w}(t_{0})$}
\STATE{Node $w$ replies with $r_{j}$ and $ID$}
\STATE{$\forall u \in Path(w,i)$, store $r_{j}, d_{u}(t_{0},j)$}
\FOR{$u_{k} \in Path(i,w), k = 1:Length(Path(i,w))$}
\STATE{$\mathcal{E}_{u_{k}}(t_{0},j) = \mathcal{E}_{u_{k}}(t_{0},j) - header$}
\STATE{Run \emph{Cost-Reward Caching} algorithm}
\IF {Caching Decision = TRUE}
\STATE{$header = header + \mathcal{E}_{u_{k}}(t_{0},j)$}
\ENDIF
\ENDFOR
\ENDIF
\end{algorithmic}
\end{algorithm}

For example, Figure \ref{fig:Prac} shows a simple network where a content $\beta_{1}$ is originally stored at $S_{1}$. We removed the triangles representing the set of non-caching nodes for the sake of clarity.  If a request for $\beta_{1}$ appears at $v_{0}$, \txtblue{node $v_{0}$ will send a control message up to the first node caching $\beta_{1}$, which is $S_{1}$, and retrieves the values of $r_{1}$ and $d_{0}(t_{0},1) = 1$. Based on these values, if} $v_{0}$ decides to cache $\beta_{1}$, \txtblue{it will send the request for $\beta_{1}$ to its parent, which is $S_{1}$, with the header field initialized to $\mathcal{E}_{0}(t_{0},1) = 14$. Node $S_{1}$ will simply reply with a data packet containing $\beta_{1}$, and $v_{0}$ will cache $\beta_{1}$}. \txtblue{Later on,} if another request for $\beta_{1}$ appears at $v_{5}$ \txtblue{while $\beta_{1}$ is still cached at $v_{0}$, node $v_{5}$ will send a control message up to the first node caching $\beta_{1}$, which is $v_{0}$. Node $v_{0}$ sends a message containing the values of $r_{1}$ and its ID to node $v_2$. Node $v_2$ will store the value of $r_1$, sets $d_{2}(t_{0},j) = 1$, and forwards the message to $v_5$. Node $v_5$ in turn will store the value of $r_1$ and set $d_{5}(t_{0},j) = 2$. Based on these values, if} $v_{5}$ decides to cache $\beta_{1}$ \txtblue{it will send the request for $\beta_{1}$ to its parent, which is $v_{2}$, with a header field initialized to $\mathcal{E}_{5}(\tau,1) = 2$. When the request reaches $v_{2}$, it will first subtract the value in the header field from its own $\mathcal{E}_{2}(\tau,1)$, so the new value of $\mathcal{E}_{2}(\tau,1)$ is $\mathcal{E}_{2}(\tau,1) = \mathcal{E}_{2}(\tau,1) - header = 4 - 2 = 2$. The reason that node $v_{2}$ has to subtract the header field from its own $\mathcal{E}_{2}(\tau,1)$ is because the requests for $\beta_{1}$ coming from the subnetwork connected to node $v_{5}$ will not be served from the cache of node $v_{2}$ since $v_{5}$ decided to cache $\beta_{1}$. Based on these values, if $v_{2}$ decides \emph{not} to cache $\beta_{1}$, it will add 0 to the header field and forward the request to its parent $v_{0}$. Node $v_{0}$ will simply reply with a data packet containing $\beta_{1}$, and only $v_{5}$ will cache $\beta_{1}$.}

\begin{figure}[h]
\centering
\includegraphics[scale = 0.8]{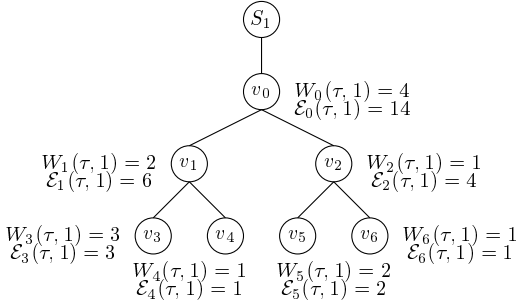}
\caption{Simple Caching Netwrok 2.}
\label{fig:Prac}
\end{figure}

The core idea of the Cost-Reward Caching algorithm is to assign an exponential cost function for each node in terms of the node's relative load. If the cost of caching a content is less than the traffic savings achieved by caching the content, the algorithm decides to cache. The choice of an exponential cost function guarantees that the node's capacity constraints are not violated. We show that in the next Section.

We define the cost of caching at a node $i$ at time $\tau$ as:
\begin{displaymath}
C_{i}(\tau,j) = D_{i}(\mu^{\lambda_{i}(\tau,j)} - 1),
\end{displaymath}
where $\mu$ is a constant defined in Section \ref{sec:Performance}. The algorithm for Cost-Reward Caching is presented in Algorithm~\ref{alg:alg2}.

\begin{algorithm}
\caption{Cost-Reward Caching (CRC)}
\label{alg:alg2}
\begin{algorithmic}
\STATE{New request for $\beta_{j}$ arriving at node $i$ at time $t_{0}$}
\STATE{$\forall \tau \in \{t_{0},\ldots,t_{0}+T_{j}(t_{0})\}$, Compute $\lambda_{i}(\tau,j), C_{i}(\tau,j)$}
\STATE{}
\IF{$\sum_{\tau = t_{0}}^{t_{0} + T_{j}(t_{0})}\mathcal{E}_{i}(\tau,j)d_{i}(t_{0},j) \geq \sum_{\tau = t_{0}}^{t_{0} + T_{j}(t_{0})}\frac{r_{j}}{D_{i}}C_{i}(\tau,j)$}
\STATE{}
\STATE{Cache $\beta_{j}$ on node $i$}
\STATE{$\tau_{0}(i,j) = t_{0}(i,j)$}
\STATE{$\forall \tau \in \{t_{0},\ldots,t_{0}+T_{j}(t_{0})\}, \lambda_{i}(\tau,j+1) = \lambda_{i}(\tau,j) + \frac{r_{j}}{D_{i}}$}
\ELSE
\STATE{Do not cache}
\ENDIF
\end{algorithmic}
\end{algorithm}
In the algorithm, when new content that is \emph{not currently cached by node $i$} arrives at time $t_{0}$, node $i$ computes the relative load ($\lambda_{i}(\tau,j)$) and the cost ($C_{i}(\tau,j)$) for every $\tau \in \{t_{0},\ldots,t_{0}+T_{j}(\tau)\}$. This is because a currently cached content may be flushed before $t_{0}+T_{j}(t_{0})$, thus the relative load and the cost should be adjusted for each time slot thereafter.

For example, Figure \ref{fig:figure1} shows the relative load at a node for the next 10 time slots starting from $t_{0}$, which is the arrival time of a new content $\beta_{4}$. The node has three cached contents, $\beta_{1}$, $\beta_{2}$, and $\beta_{3}$ that are going to be flushed at times $\tau_{1} = t_{0}+3$, $\tau_{2} = t_{0}+9$, and $\tau_{3} = t_{0}+7$, respectively. When a $\beta_{4}$ arrives at this node at $\tau = t_{0}$ with $T_{4}(t_{0}) = 10$, the cost calculation should include three cached contents for 3 time slots, two cached contents for 4 time slots, one cached content for 2 time slots, and 0 cached content for 1 time slot. If the total savings for caching $\beta_{4}$ is greater than the aggregated cost, then $\beta_{4}$ will be cached on node $i$, and the relative load is updated to include the effect of $\beta_{4}$.

\begin{figure}[h]
\centering
\includegraphics[scale = 0.5]{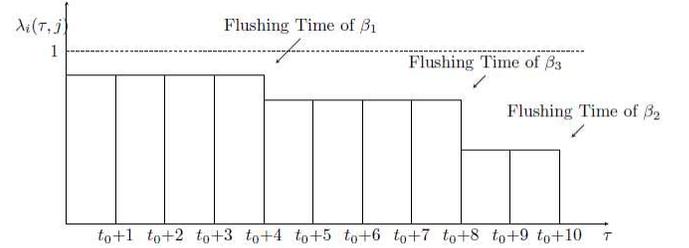}
\caption{Relative Load Calculation Example. The figure shows the state of the cache in one node when it considers a new content $\beta_{4}$ for caching at time $t_{0}$ and $T_{4}(t_{0}) = 10$. We have three contents, $\beta_{1}$, $\beta_{2}$, and $\beta_{3}$, that are to be flushed at times $\tau_{1} = t_{0}+3$, $\tau_{2} = t_{0}+9$, and $\tau_{3} = t_{0}+7$, respectively.}
\label{fig:figure1}
\end{figure}

\subsection{Practical Issues}\label{subsec:Practical}
So far, we developed a fully distributed algorithm that achieves asymptotic optimality in terms of traffic savings under some realistic assumptions. Before providing the optimality proof, we discuss in this section the practical issues that make the algorithm easy to implement. The major issues in our algorithm include providing incentives for the caching nodes and QoS guarantees for the content providers, the adoption of en-route caching, calculating the popularity expectation of each content, and updating the effective caching duration.

\subsubsection{Providing Incentives and QoS Guarantees}\label{subsec:qos}
In this work, the QoS measure is to guarantee the existence of the content in the cache for a certain period of time, so the content will be delivered quickly. In other words, once a caching node decides to cache a certain content, the content will not be replaced during the effective caching time of the content. Providing such a guarantee along with adopting an equal pay charging policy for all contents will provide the caching nodes with the necessary incentive to cache. Figure \ref{fig:qos} shows the interaction between the ISP and the content provider.

\begin{figure}[h]
\centering
\includegraphics[scale = 0.4]{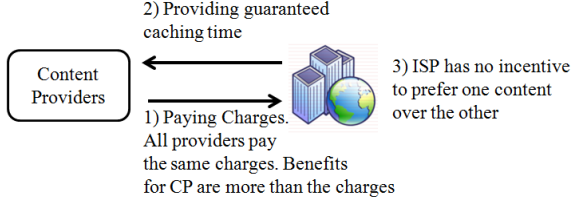}
\caption{Interaction between ISP and Content Provider.}
\label{fig:qos}
\end{figure}

We assume that the caching nodes should adopt charging policies, where every content provider is charged the same. This will prevent the caching node from preferring one content over the other. Moreover, such charging policies will enforce the caching nodes to cooperate and apply our CRC algorithm.

\subsubsection{En-Route Caching}
In en-route caching, a request for $\beta_{j}$ will be sent to the parent along the traditional path to the content's source, until the request reaches the first node caching the content or the content's source. The adoption of this en-route caching reduces the amount of broadcasted \emph{Interest} packets as opposed to the currently deployed schemes in CCN, where the interest packets are broadcasted to all neighbors. Moreover, using en-route caching prevents the reception of multiple copies of the requested content as opposed to CCN. Furthermore, our algorithm can be easily implemented in the current Internet architecture.

\subsubsection{Calculating the Initial Content Expectation Values}\label{subsec:expec}
For each content, we start by building a caching tree rooted at the source of the content. The caching tree is the union of the traditional paths from the source of the content to all other nodes. We calculate the initial expectation value at a caching node for a certain content, when only node $S_j$ holds the $j$-th content, based on the content's popularity and the number of end nodes in the subnetwork connected to that node. For example, in Figure \ref{fig:Simple}, $W_{3}(\tau,j)$ at node $v_{3}$ for content $\beta_{j}$ is proportional to the content's popularity and the number of end nodes in the subnetwork connected to node $v_{3}$.

Algorithm \ref{alg:alg3} shows how to calculate $\mathcal{E}_{i}(\tau,j)$ for each content at each caching node before the appearance of any request at any node. The expectations are calculated in a distributed way, where each node only needs to know the expectation values of its children in the caching tree. In the simulation, we investigate the effect of having error margins in the expectation calculation.

\begin{algorithm}
\caption{Initial Content Popularity Expectation Calculation}
\label{alg:alg3}
\begin{algorithmic}
\FOR{each content $\beta_{j} = \{S_{j},r_{j},T_{j}(\tau)\}$}
\STATE{$CachingTree(j) \leftarrow$ build the traditional path tree rooted at $S_{j}$}
\FOR{each caching node $i \in CachingTree(j)$}
\STATE{Calculate $W_{i}(\tau,j)$}
\STATE{Initialize $\mathcal{E}_{i}(\tau,j) \leftarrow W_{i}(\tau,j)$}
\ENDFOR
\FOR{each node $z \in Ancestor(i)$ in $CachingTree(j)$}
\STATE{$\mathcal{E}_{z}(\tau,j) = \mathcal{E}_{z}(\tau,j) + W_{i}(\tau,j)$}
\ENDFOR
\ENDFOR
\STATE{}
\end{algorithmic}
\end{algorithm}

For example, referring back to Figure \ref{fig:Prac}, and before a request for $\beta_{1}$ appears at any node, the values of $\mathcal{E}_{i}(\tau,j)$ are calculated as described in Algorithm \ref{alg:alg3}. Take node $v_{2}$ for example, then $\mathcal{E}_{2}(\tau,1) = W_{2}(\tau,1) + W_{5}(\tau,1) + W_{6}(\tau,1) = 4$. The final expectation values for the rest of the nodes are shown in the figure.

\subsubsection{Effective Caching Duration}\label{subsec:Timing}
The effective caching duration of a content depends on its arrival time. For example, most people read the newspaper in a period of two hours, so the caching duration should be two hours beginning at the arrival of the first request. However, if a new request for the newspaper arrives at a node in the middle of the range and was cached by the algorithm, then the caching duration should be one hour. This requires the broadcast of the first arrival time to all other nodes in the network. The additional overhead incurred by such broadcasting is negligible compared to the reduction of the \emph{Interest} packet broadcasting we achieve through the adoption of en-route caching.

\section{Performance Analysis}\label{sec:Performance}
In this Section, we show that any online algorithm has a competitive ratio that is lower bounded by $\Omega(\log(n))$, then we show that our algorithm does not violate the capacity constraints, and achieves a competitive ratio that is upper bounded by $\mathcal{O}(\log (n))$ under realistic settings.

\begin{proposition}\label{prop:Proposition2}
Any online algorithm has a competitive ratio which is lower bounded by $\Omega(\log(n))$.
\end{proposition}

\begin{proof}
We show this proposition by giving an example network, such that the best online algorithm competitive ratio is lower bounded by $\Omega(\log(n))$. Consider a network which consists of $n+2$ nodes, as shown in Figure \ref{lb}. All contents are originally placed at node $S$, and node $C$ is the only node with caching capability with a unit cache capacity. All other nodes can request the contents. We consider a 2-time slots system where all contents are to be requested at the beginning of each time slot, though sequentially. Sequentially means that the algorithm has to make a caching decision for a content before considering the next one.

\begin{figure}[h]
\Huge
\centering
\includegraphics[scale = 0.7]{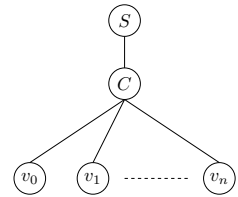}
\caption{Network for Lower Bound Proof.}
\label{lb}
\end{figure}

Consider a $\log(n) + 1$ phases of contents. For each phase $0 \leq i \leq \log(n)$, we have $1/\alpha$ identical contents, each with size $\alpha \ll 1$ and a caching time equal to 2 time slots. Contents in the same phase are destined for the same $2^{i}$ nodes.
The reason behind considering a 2-time slots system is that when a node caches a content, the traffic saving is considered for future requests.

Let $x_{i}$ be the fraction of contents stored from phase $i$ and $G_{i}$ be the traffic saving of the online algorithm gained from phase $i$, then
\begin{displaymath}
G_{i} = x_{i}2^{i}
\end{displaymath}
Consider the first $k$ phases, then the online traffic saving of these $k$ phases, denoted by $G(k)$, is
\begin{displaymath}
G(k) = \sum G_{i} = \sum_{i=0}^{k}x_{i}2^{i}
\end{displaymath}
The offline algorithm will cache the contents from phase $k$ only, gaining a traffic saving of $2^{k}$

Now consider the ratio of the online traffic saving to the offline traffic saving:
\begin{align}
\sum_{k=0}^{\log n}\frac{G(k)}{2^{k}} &= \sum_{k=0}^{\log n}\sum_{i=0}^{k}\frac{x_{i}2^{i}}{2^{k}} = \sum_{i=0}^{\log n}\sum_{k=i}^{\log n}x_{i}2^{i-k}\nonumber \\
&= \sum_{i=0}^{\log n}x_{i}\sum_{k=i}^{\log n}2^{i-k} \leq 1*2 \leq 2 \nonumber
\end{align}
Hence, there exist some $k$ such that $\frac{G(k)}{2^{k}} \leq \frac{2}{\log n}$. This means that the saving of the offline algorithm is at least within a $\log n$ factor of the savings achieved by any online algorithm.
\end{proof} 

Before we start the proof of satisfying the capacity constraints and the upper bound, we need to state the following two assumptions:
\begin{equation}
1 \leq \frac{1}{n}. \frac{\mathcal{E}_{i}(\tau,j)b_{i}(j)}{r_{j}T_{j}(\tau)} \leq F  \qquad \forall j, \forall i\neq S_{j}, \forall \tau \label{A1},
\end{equation}
and
\begin{equation}
r_{j} \leq \frac{\min{D_{i}}}{\log(\mu)} \qquad \forall j \label{A2},
\end{equation}
where $F$ is any constant large enough to satisfy the assumption in \eqref{A1}, $\mu = 2(nTF + 1)$, $n$ is the number of caching nodes, and $T = \max(T_{j}), \forall j$. The assumption in \eqref{A1} states that the amount of traffic savings for a content scales with the content's size and caching duration. The assumption in \eqref{A2} requires that the caching capacity of any node should be greater than the size of any content, which is a practical condition to assume.

We start by proving that the CRC algorithm does not violate the capacity constraints. After that, we show that CRC achieves a $\mathcal{O}(\log (n))$ competitive ratio. In all of the subsequent proofs, $\tau \in \{t_{0}(i,j),\ldots,t_{0}(i,j)+T_{j}(t_{0}(i,j))\}$, where $t_{0}(i,j)$ is the arrival time of $\beta_{j}$ at node $i$.

\begin{proposition}
The CRC algorithm does not violate the capacity constraints.
\end{proposition}

\begin{proof}
Let $\beta_{j}$ be the first content that caused the relative load at node $i$ to exceed 1. By the definition of the relative load, we have
\begin{displaymath}
\lambda_{i}(\tau,j) > 1 - \frac{r_{j}}{D_{i}}
\end{displaymath}
using the assumption in \eqref{A2} and the definition of the cost function, we get

\begin{align}
\frac{C_{i}(\tau,j)}{D_{i}} &= \mu^{\lambda_{i}(\tau,j)} - 1 \geq \mu^{1 - \frac{r_{j}}{D_{i}}} - 1 \nonumber\\
& \geq \mu^{1 - \frac{1}{\log \mu}} - 1 \geq \frac{\mu}{2} - 1 \geq nTF \nonumber
\end{align}
Multiplying both sides by $r_{j}$ and using the assumption in \eqref{A1}, we get
\begin{center}
$\frac{r_{j}}{D_{i}}C_{i}(\tau,j) \geq nTFr_{j} \geq \mathcal{E}_{i}(\tau,j)b_{i}(j) \geq \mathcal{E}_{i}(\tau,j)d_{i}(t_{0},j)$
\end{center}
From the definition of our algorithm, $\beta_{j}$ should not be cached at node $i$. Therefore, the CRC algorithm does not violate the capacity constraints.
\end{proof}

The next lemma shows that the traffic saving gained by our algorithm is lower bounded by the sum of the caching costs.
\begin{lemma}
Let $A$ be the set of indices of contents cached by the CRC algorithm, and $k$ be the last index, then
\begin{equation}
2 \log(\mu) \sum_{i,j \in A,\tau} [\mathcal{E}_{i}(\tau,j)d_{i}(t_{0},j)] \geq \sum_{i,\tau}C_{i}(\tau,k+1) \label{lemma1}
\end{equation}
\label{lemma:lemma1}
\end{lemma}

\begin{proof}
By induction on $k$. When $k = 0$, the cache is empty and the right hand side of the inequality is 0. When $\beta_{j}$ is not cached by the online algorithm, neither side of the inequality is changed. Then it is enough to show, for a cached content $\beta_{j}$, that:

\begin{eqnarray}
\lefteqn{2 \log(\mu) \sum_{i,\tau}[\mathcal{E}_{i}(\tau,j)d_{i}(t_{0},j)]} \nonumber \\
 & &\geq \sum_{i,\tau}[ C_{i}(\tau,j+1) - C_{i}(\tau,j)] \nonumber
\end{eqnarray}
since summing both sides over all $j \in A$ will yield \eqref{lemma1}.

Consider a node $i$, the additional cost incurred by caching $\beta_{j}$ is given by:
\begin{align}
C_{i}(\tau,j+1) - C_{i}(\tau,j) &= D_{i}[\mu^{\lambda_{i}(\tau,j+1)} - \mu^{\lambda_{i}(\tau,j)}] \nonumber\\
&=  D_{i}\mu^{\lambda_{i}(\tau,j)}[\mu^{\frac{r_{j}}{D_{i}}} - 1] \nonumber\\
&= D_{i}\mu^{\lambda_{i}(\tau,j)}[2^{\log\mu \frac{r_{j}}{D_{i}}} - 1] \nonumber
\end{align}

Since $2^{x} - 1 \leq x$ for $0 \leq x \leq 1$ and using the assumption in \eqref{A2}

\begin{align}
C_{i}(\tau,j+1) - C_{i}(\tau,j) &\leq D_{i}\mu^{\lambda_{i}(\tau,j)}[\frac{r_{j}}{D_{i}}\log\mu] \nonumber\\
&\leq r_{j}\log\mu [\frac{C_{i}(\tau,j)}{D_{i}} + 1] \nonumber\\
&\leq \log\mu [\frac{ r_{j}}{D_{i}}C_{i}(\tau,j) + r_{j}] \nonumber
\end{align}
Summing over $\tau$, $i$, and the fact that $\beta_{j}$ is cached, we get
\begin{eqnarray}
\lefteqn{\sum_{i}\sum_{\tau}[ C_{i}(\tau,j+1) - C_{i}(\tau,j)]} \nonumber\\
& &\leq  \log\mu \sum_{i}\sum_{\tau}[\frac{ r_{j}}{D_{i}}C_{i}(\tau,j) + r_{j}] \nonumber\\
& &\leq \log\mu [\sum_{i}\mathcal{E}_{i}(\tau,j)d_{i}(t_{0},j) + \sum_{i}\sum_{\tau} r_{j}] \nonumber\\
& &\leq 2 \log\mu \sum_{i}\mathcal{E}_{i}(\tau,j)d_{i}(t_{0},j) \nonumber
\end{eqnarray}
\end{proof}

In the next lemma, $d_{i}(\tau,j)$ is defined for the online algorithm.

\begin{lemma}
Let $Q$ be the set of indices of contents cached by the offline algorithm, but not the CRC algorithm. Let $l = \argmax_{j \in Q}(C_{i}(\tau,j))$. Then
\begin{displaymath}
\sum_{i}\sum_{j \in Q}\sum_{\tau} [\mathcal{E}_{i}(\tau,j)d_{i}(t_{0},j)] \leq \sum_{i}\sum_{\tau} C_{i}(\tau,l)
\end{displaymath}
\label{lemma:lemma2}
\end{lemma}

\begin{proof}
Since $\beta_{j}$ was not cached by the online algorithm, we have:
\begin{align}
\sum_{\tau}\mathcal{E}_{i}(\tau,j)d_{i}(t_{0},j) &\leq \sum_{\tau}\frac{r_{j}}{D_{i}}C_{i}(\tau,j) \nonumber \\
& \leq \sum_{\tau}\frac{r_{j}}{D_{i}}C_{i}(\tau,l) \nonumber \\
\sum_{i}\sum_{\tau}\mathcal{E}_{i}(\tau,j)d_{i}(t_{0},j) &\leq \sum_{i}\sum_{\tau}\frac{r_{j}}{D_{i}}C_{i}(\tau,l) \nonumber
\end{align}
Summing over all $j \in Q$
\begin{align}
\sum_{i}\sum_{j \in Q}\sum_{\tau}\mathcal{E}_{i}(\tau,j)d_{i}(t_{0},j) &\leq \sum_{i}\sum_{\tau}C_{i}(\tau,l)\sum_{j \in Q}\frac{r_{j}}{D_{i}} \nonumber \\
& \leq \sum_{i}\sum_{\tau}C_{i}(\tau,l) \nonumber
\end{align}

Since any offline algorithm cannot exceed a unit relative load, $\sum_{j \in Q}\frac{r_{j}}{D_{i}} \leq 1$.
\end{proof}

Combining Lemma \ref{lemma:lemma1} and Lemma \ref{lemma:lemma2}, we have the following lemma.
\begin{lemma}
Let $A^{*}$ be the set of indices of the contents cached by the offline algorithm, and let $k$ be the last index. Then:
\begin{eqnarray}
\lefteqn{\sum_{i,j \in A^{*},\tau}\mathcal{E}_{i}(\tau,j)d_{i}(t_{0},j)}\nonumber \\
& &\leq 2\log(2\mu) \sum_{i,j \in A,\tau}\mathcal{E}_{i}(\tau,j)d_{i}(t_{0},j)\nonumber
\end{eqnarray}
\end{lemma}

\begin{proof}
The traffic savings of the offline algorithm is given by:
\begin{eqnarray}
\lefteqn{\sum_{i,j \in A^{*},\tau}\mathcal{E}_{i}(\tau,j)d_{i}(t_{0},j)}\nonumber \\
& &=\sum_{i,j \in Q,\tau}\mathcal{E}_{i}(\tau,j)d_{i}(t_{0},j) + \sum_{i,j \in A^{*}/Q,\tau}\mathcal{E}_{i}(\tau,j)d_{i}(t_{0},j) \nonumber \\
& &\leq \sum_{i,j \in Q,\tau}\mathcal{E}_{i}(\tau,j)d_{i}(t_{0},j) + \sum_{i,j \in A,\tau}\mathcal{E}_{i}(\tau,j)d_{i}(t_{0},j) \nonumber \\
& &\leq \sum_{i,\tau}C_{i}(\tau,l) + \sum_{i,j \in A,\tau}\mathcal{E}_{i}(\tau,j)d_{i}(t_{0},j) \nonumber \\
& &\leq \sum_{i,\tau}C_{i}(\tau,k+1) + \sum_{i,j \in A,\tau}\mathcal{E}_{i}(\tau,j)d_{i}(t_{0},j) \nonumber \\
& &\leq (2\log\mu + 1)\sum_{i,j \in A,\tau}\mathcal{E}_{i}(\tau,j)d_{i}(t_{0},j) \nonumber \\
& &\leq 2\log(2\mu) \sum_{i,j \in A,\tau}\mathcal{E}_{i}(\tau,j)d_{i}(t_{0},j) \nonumber
\end{eqnarray}
\end{proof}

Note that $d_{i}(\tau,j)$ in the previous lemmas is defined by the online algorithm. In order to achieve optimality using this proof technique, $d_{i}(\tau,j)$ of the online algorithm should be equal to $d_{i}(\tau,j)$ of the offline algorithm. In the next two corollaries, we show cases where $d_{i}(\tau,j)$ of the online algorithm is equal to $d_{i}(\tau,j)$ of the offline algorithm.  

\begin{corollary}
When there is only one caching node in every path, then $d_{i}(\tau,j)$ of the online algorithm is equal to $d_{i}(\tau,j)$ of the offline algorithm, and our algorithm achieves asymptotic optimality.
\label{cor:corollary1}
\end{corollary}

\begin{corollary}
When every node in the path shares the same caching decision, then $d_{i}(\tau,j)$ of the online algorithm is equal to $d_{i}(\tau,j)$ of the offline algorithm, and our algorithm achieves asymptotic optimality.
\label{cor:corollary2}
\end{corollary}

\section{Extension to CRC Algorithm}\label{sec:Extensions}
In this section, we provide an extension to the CRC algorithm. We show the effeciency of this extension with respect to currently deployed caching schemes through extensive simulations.

\subsection{Replacement-CRC}
The basic CRC algorithm provides quality of service guarantees for content providers by not replacing their contents once they are cached. Content providers, in return, are chraged to provide incentives for the caching nodes based on the caching policy discussed in section \ref{subsec:qos}. In this section, we present an extension for the basic CRC algorithm that allows content replacement.

The settings for Replacement-CRC are the same as for the basic CRC algorithm. However, there is no restriction on keeping a content $\beta_{j}$ in the cache of node $i$ for the whole effective caching duration time $T_{j}(\tau)$, as $\beta_{j}$ may be replaced by another content.

We present the details of the Replacement-CRC algorithm in algorithm \ref{alg:alg5}.

\begin{algorithm}
\caption{Replacement-CRC}
\label{alg:alg5}
\begin{algorithmic}
\STATE{A new request for $\beta_{j}$ appears at node $i$ at time $t_{0}$}
\STATE{$\forall \tau \in \{t_{0},\ldots,t_{0}+T_{j}(t_{0})\}$, Compute $\lambda_{i}(\tau,j)$, $C_{i}(\tau,j)$}
\IF{$\sum_{\tau}\mathcal{E}_{i}(t_{0},j)d_{i}(t_{0},j) \geq \sum_{\tau}\frac{r_{j}}{D_{i}}C_{i}(\tau,j)$}
\STATE{Cache $\beta_{j}$ at node $i$}
\STATE{$\tau_{0}(i,j) = t_{0}(i,j)$}
\STATE{$\forall \tau \in \{t_{0},\ldots,t_{0}+T_{j}(t_{0})\}, \lambda_{i}(\tau,j+1) = \lambda_{i}(\tau,j) + \frac{r_{j}}{D_{i}}$}
\ELSE 
\STATE{$\forall \beta_{k} \in Cache_{i}(t_{0}) \cup \beta_{j}, \forall \tau \in \{t_{0},\ldots,t_{0}+T_{k}(t_{0})\}$, Compute}
\STATE{$\lambda_{i}^{k}(\tau,j) = \lambda_{i}(\tau,j) + \frac{r_{j}}{D_{i}} - \frac{r_{k}}{D_{i}}$}
\STATE{$C_{i}^{k}(\tau,j) = D_{i}[\mu^{\lambda_{i}^{k}(\tau,j) - 1}]$}
\IF{$\lambda_{i}^{k}(\tau,j) \leq 1$}
\STATE{$Diff(k) = \sum_{\tau}\mathcal{E}_{i}(\tau_{0},k)d_{i}(\tau_{0},k) - \sum_{\tau}\frac{r_{j}}{D_{i}}C_{i}^{k}(\tau,j)$}
\ENDIF
\STATE{$l = \argmin_{k}(Diff)$}
\IF{$l \neq j$}
\STATE{Replace $\beta_{l}$ with $\beta_{j}$}
\STATE{$\forall \tau \in \{t_{0},\ldots,t_{0}+T_{j}(t_{0})\}, \lambda_{i}(\tau,j+1) = \lambda_{i}^{l}(\tau,j)$}
\ENDIF
\ENDIF
\end{algorithmic}
\end{algorithm}

Algorithm \ref{alg:alg5} states that if the traffic savings gained by caching a new content $\beta_{j}$ is greater than the caching cost at node $i$, then the algorithm decides to cache. Otherwise, we compare the difference between the traffic savings and the caching costs for every $\beta_{k} \in Cache_{i}(\tau)$, if it is replaced by $\beta_{j}$ without violating the capacity constraints. We then choose the content with the minimum difference to replace with $\beta_{j}$.

\section{Simulation Results}\label{sec:Simulation}
In this Section, we compare our CRC algorithm to some of the existing caching schemes.
\subsection{Settings}\label{subsec:SimSettings}
We simulate the following caching schemes:

{\bf(1)} CRC: This scheme represents our basic algorithm.

{\bf(2)} CRC Version 2: This is similar to the CRC scheme, Version 1, except that we retrieve the content from the closest node that has the content in its cache, not necessarily along the path to the content's source.

{\bf(3)} All Cache: This scheme caches every new content arriving at a caching node, as long as there is enough residual capacity to cache the new content.

{\bf(4)} Random Caching Version 1: In this scheme, when a request for a content arrives at node $i$, the caching probability of the content depends on the content's popularity at node $i$. The popularity of a content $\beta_{j}$ at node $i$ denoted by $Pop_{j}$, is defined as the ratio of the number of requests for $\beta_{j}$ coming from the subnetwork connected to node $i$ denoted by $N_{i}^{j}$, to the total number of non-caching nodes in the subnetwork connected to node $i$ denoted by $N_{i}$. Mathematically speaking, $Pop_{j} = N_{i}^{j}/N_{i}$. If we choose a uniform random number $x$ between [0,1], and $x \leq Pop_{j}$, then the content $\beta_{j}$ is cached if there is enough room for it in the cache. Otherwise, the content is not cached.

{\bf(5)} Random Caching Version 2: This is similar to Random Caching Version 1, except that the caching probability of the content depends on the content's popularity at node $i$, scaled by the fraction of the available residual capacity to the total capacity in the cache of node $i$ denoted by $f_i$, \emph{i.e.}, if we choose a uniform random number $x$ between [0,1], and $x \leq f_i \times Pop_{j}$, then the content $\beta_{j}$ is cached if there is enough room for it in the cache. Otherwise, the content is not cached.

For every caching node $i$ in the network, we assign a cache capacity $D_{i}$ that is uniformly chosen in the range of [750,1000] GB. The number of the non-caching nodes connected to the caching node $i$ is chosen uniformly at random in the range of 10 to 90 nodes.

For every content, we randomly chose one of the nodes to act as the source. Each content has a size chosen randomly in the range of [100,150] MB. The starting effective time of the content is chosen randomly. The end time is also chosen randomly within a fixed interval from the starting time. If the end time exceeds the end time of the simulation, it is adjusted to be equal to the end time of the simulation. The simulation interval is chosen to be 1000 time slots.

\subsection{Results on Random topologies}
We start our evaluation on random backbone topologies, in which the caching nodes are generated as a random topology.

We simulate the effect of the number of caching nodes $n$ in the network for three cases, $n = 30$, $n = 50$, and $n = 100$ nodes. For each case we use 10 random topologies, and report the average performance. We fix the effective caching duration to 150 slots and the number of contents to 10000 contents to solely show the effect of increasing the number of nodes on the performance of the CRC algorithm. The results are shown in Figure \ref{fig:RA}(a). 

As can be seen from the figure, increasing the number of the caching nodes will result in better performance in all schemes since more contents can be cached. Another observation from the figure is that the performance of CRC schemes increases at a higher rate than other schemes as we increase the number of the nodes in the network. This shows that our scheme greatly benefits from adding more caching nodes to the network. It is also aligned with the property of asymptotic optimality of our scheme. On the other hand, not much improvement can be seen from the other schemes when the number of nodes is increased in the network.

We simulate the effect of changing the number of contents from 2000 to 10000. The results are averaged over 10 runs and are shown in Figure \ref{fig:RA}(b). The reason that the performance of the Cache All, Random 1, and Random 2 schemes increases and then decreases is that there is a saturation point after which the caches of the network cannot handle the requests. On the other hand, Our scheme reserves the cache capacity for contents with higher traffic savings, and achieves an improvement of 2 to 3-fold in terms of traffic savings.

Figure~\ref{fig:RA}(c) shows the effect of the maximum effective caching duration for three cases, 50, 100, and 150 time slots. In this scenario, the difference between the start and end times for each content is drawn randomly from $\{1, \ldots, \max. caching\ duration\}.$ The reason that the traffic savings decrease as the maximum effective caching duration increases after a certain point is that contents are cached for a longer period, so future contents are less likely to find enough residual capacity at the caching node.

In all of the results in Figure \ref{fig:RA}, the performance of CRC Version 2 is always less than the performance of CRC Version 1. This is because CRC Version 2 deviates from the settings under which we achieve optimality.

\begin{figure*}
\centering
\includegraphics[scale = 0.5]{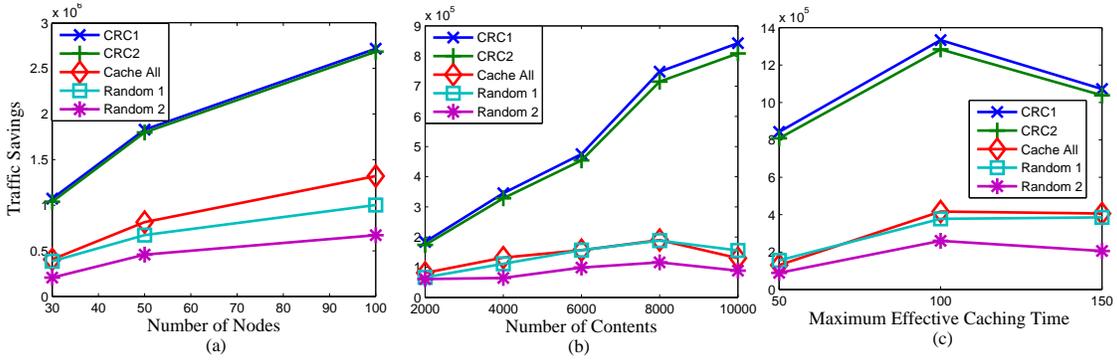}
\caption{The Effects of Different Factors on the Performance of the Random Topologies.}
\label{fig:RA}
\end{figure*}

So far our performance measure was the traffic saving. In Figure~\ref{fig:cost}, we measure the cost in terms of total number of hops to satisfy all of the requests. The results in Figure~\ref{fig:cost} are for a random topology with 100 caching nodes, the number of contents is 10000, and the maximum effective caching duration is 150 slots. The results in the figure show that even when we measure the performance in terms of the total cost, our scheme reduces the cost by the range of 30\% to 50\%.

\begin{figure}[t]
\centering
\includegraphics[scale = 0.3]{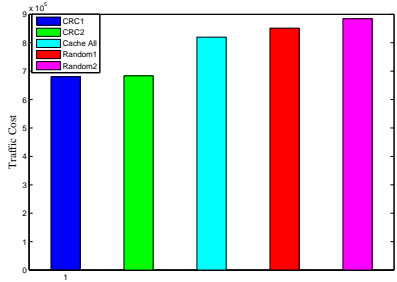}
\caption{Traffic cost.}
\label{fig:cost}
\end{figure}

In Figure~\ref{fig:CDF} we measure the per topology improvement for all schemes with respect to Random Caching Version 2 scheme. Here, we measure the performance of all schemes for 100 different random topologies. For each topology, we normalize the performance of all schemes with respect to the performance of Random Caching Version 2. Denote the performance of the CRC scheme and Random Caching Version 2 scheme for topology $s$ as $P_{CRC}(top_{s})$ and $P_{Random2}(top_{s})$, respectively. We compute the normalized performance of CRC scheme with respect to Random Caching Version 2 scheme for topology $s$ as $R_{CRC}(top_{s}) = P_{CRC}(top_{s})/P_{Random2}(top_{s})$. After that, the empirical CDF of the vector $R_{CRC} = [R_{CRC}(top_{1}), R_{CRC}(top_{2}), \ldots , R_{CRC}(top_{100})]$ for the 100 random topologies is plotted. We do the same process for the other two schemes. The results in the figure show that our scheme experiences about 4 times the improvements as that by Random Caching Version 2.

\begin{figure}
\centering
\includegraphics[scale = 0.4]{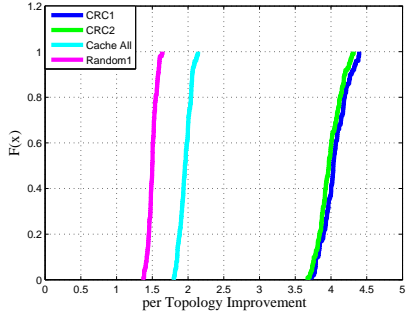}
\caption{The empirical CDF of the per topology improvement for random topologies with respect to Random Caching Version 2.}
\label{fig:CDF}
\end{figure}

\subsection{Results on a Small-word generated topology}
In \cite{bu2002distinguishing} it is shown that the Internet topology exhibits a small-world structure defined in \cite{watts1998collective}. In this Section we perform simulations based on the small world-structure.

Figure \ref{fig:SA} is similar to Figure \ref{fig:RA}, but for the small-world topologies. The results follow the same trend as the results for the random topologies except for two differnces. The first difference is that is that CRC Version 1 achieves better perofrmance than CRC Version 2 as we increase the number of nodes. The second difference is that all of the schemes performances increase with increasing the effective caching time. One of the reason is due to the sparsity of the small-world topologies, which results in the fact that the requests are distributed over multiple domains inside the topology.

\begin{figure*}
\centering
\includegraphics[scale = 0.5]{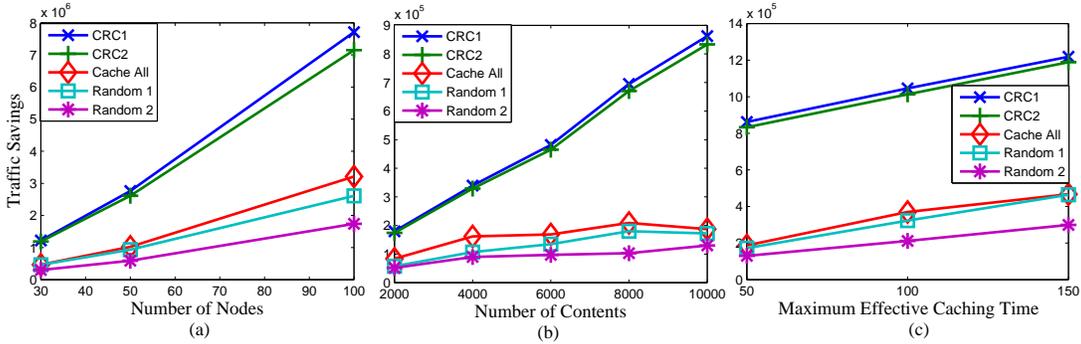}
\caption{The Effects of Different Factors on the Performance of the Small-world Topologies.}
\label{fig:SA}
\end{figure*}

\subsection{Results for Replacement-CRC}
We compare the performance of Replacement-CRC against the following schemes:

{\bf(1)} Least Recently Used (LRU): In this scheme, when a request for a content $\beta_{j}$ appears at node $i$, Least Recently Used replacement is performed at all nodes along the path from node $i$ to the source of the content $\beta_{j}$.

{\bf(2)} Random Replacement: In this scheme, when a request for a content $\beta_{j}$ appears at node $i$, every node along the path from node $i$ to the source of the content $\beta_{j}$ will randomly choose a cached content to be replaced with $\beta_{j}$, as long as the capacity constraints are satisfied.

{\bf(3)} CCN: This scheme represents the Content Centric Network as described in \cite{jacobson2009networking}, where a request for a content is broadcasted until the closest node with a copy of the content in its cache is found. The content then follows the path from the closest node to the requester, and all nodes along that path caches the content as long as the capacity constraints are satisfied, or performs replacement using LRU if content replacement is needed.

We use the same settings as described in Section \ref{subsec:SimSettings}, and we simulate the effect of increasing the number of caching nodes in the network, the effect of increasing the number of contents, and the effect of increasing the cache capacity of the caching nodes. The results are shown in Figure \ref{fig:RepAll}.

\begin{figure*}
\centering
\includegraphics[scale = 0.5]{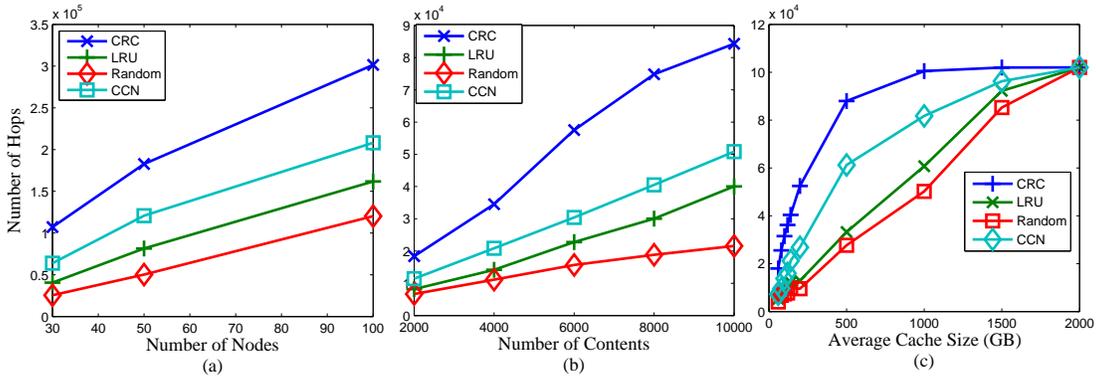}
\caption{The Effects of Different Factors on the Performance of Different Replacement Schemes.}
\label{fig:RepAll}
\end{figure*}

Figure \ref{fig:RepAll}(a) shows the performance of all schemes as we increase the number of the caching nodes in the network. From the figure, the performance of all schemes increases with increasing the number of caching nodes. This is because adding more caching nodes will increase the overall caching capacity of the network, which results in more cached contents. Moreover, as the topology grows with adding more nodes, the average distance between the nodes in the netwrok increases. The figure shows that Replacement-CRC outperforms the existing replacement schemes by 30\% to 60\%. 

Figure \ref{fig:RepAll}(b) shows the performance of all schemes as we increase the number of contents. As we increase the number of contents, the performance of all schemes increases since more contents are available for caching. Replacement-CRC acheives better perofrmance than the other schemes, since it is able to identify the contents with higher traffic savings and the replacement is done less frequently than the other schemes.

In Figure \ref{fig:RepAll}(c), we investigate the effect of increasing the caching size of the caching nodes on the performance of all schemes. We increased the caching size of each node until we reach a saturation point, where all of the nodes are able to cache all of the contents without the need for replacement. At this saturation point, all schemes acheives the same traffic savings. Another observation from the figure is that the performance of Replacement-CRC at 500GB is similar to the performance of the other schemes at 1500GB. This means that Replacement-CRC can achieve the same performance of the other schemes with only 30\% of the cache capacity.

\section{Conclusion}\label{sec:Conclusion}
Caching at intermediate nodes has the advantage of bringing the contents closer to the users, which results in traffic offloading from the origin servers and lower delays. To acheive this, caching schemes such as en-route caching and CCN have been investigated. Unlike CCN, the use of en-route caching does not require major changes to the TCP/IP model. Previous works have studied en-route caching under offline settings to acheive the optimal content placement strategy. In this work, we study the framework of en-route caching under online settings.

Under this framework, we characterize the fundamental limit for the ratio of the performance of the optimal offline scheme to that of any online scheme. The offline scheme has a complete knowledge of all of the future requests, while the online scheme does not possess such knowledge. We also design an efficient online scheme and prove that the developed online scheme achieves optimality as the number of nodes in the network becomes large. Moreover, we introduce an extension to the algorithm. Our simulation results affirm the efficiency of our scheme and its extension. Our future work includes the investigation of network coding \cite{ostovari2014network,khreishah2012distributed} under our settings.

\section*{Acknowledgement}
This research is supported by NSF grant ECCS-1331018.

\bibliographystyle{IEEEtran}
\bibliography{biblo}

\end{document}